\documentclass{toc}

\newcommand\bits{\{0,1\}}
\newcommand{\eps}{\varepsilon}
\newcommand{\ket}[1]{|#1\rangle}

\newcommand{\ketbra}[2]{|#1\rangle\langle#2|}
\newcommand{\inpc}[2]{\langle{#1},{#2}\rangle} %
\newcommand{\norm}[1]{{\left\|{#1}\right\|}}
\newcommand{\HM}{\mbox{\rm HM}}
\newcommand{\E}{\mathop{\mathbb E}}
\newcommand{\M}{\mathcal{M}}

\newcommand{\exprefsilent}[2]{\texorpdfstring{\hyperref[#2]{#1}}{#1}}

\tocdetails{%
volume=8, number=27, year=2012, firstpage=623,
received={September 15, 2011},
revised={August 31, 2012},
published={December 27, 2012},
doi={10.4086/toc.2012.v008a027}
}
\runningtitle{Near-Optimal and Explicit Bell Inequality Violations}

\tocpdftitle{Near-Optimal and Explicit Bell Inequality Violations}

\runningauthor{Harry Buhrman, Oded Regev, Giannicola Scarpa, and Ronald de Wolf}
\copyrightauthor{Harry Buhrman, Oded Regev, Giannicola Scarpa, and Ronald de Wolf}
\tocpdfauthor{Harry Buhrman, Oded Regev, Giannicola Scarpa, Ronald de Wolf}

\begin{document}

\begin{frontmatter}
\title{Near-Optimal and Explicit\\ Bell Inequality Violations\footnote{An 
earlier version of this paper appeared in the 
\href{http://dx.doi.org/10.1109/CCC.2011.30}{Proceedings of the 
26th IEEE Conference on Computational Complexity, pages 157--166, 2011.}}}

\author[buhrman]{Harry Buhrman\thanks{Supported by a Vici grant from the Netherlands Organisation for Scientific Research (NWO), and by the European Commission under the project QCS (Grant No.~255961).}}
\author[regev]{Oded Regev\thanks{%
   Supported by the Israel Science Foundation,
   by the Wolfson Family Charitable Trust, and by a European Research Council (ERC) Starting Grant.
   Part of the work done while a DIGITEO visitor in LRI, Orsay.}}
\author[scarpa]{Giannicola Scarpa\thanks{Supported by a Vidi grant from the Netherlands Organisation for Scientific Research (NWO), and by the European Commission under the project QCS (Grant No.~255961).}}
\author[dewolf]{Ronald de Wolf\thanks{Supported by a Vidi grant from the Netherlands Organisation for Scientific Research (NWO), and by the European Commission under the project QCS (Grant No.~255961).}}

\begin{abstract}
Entangled quantum systems can exhibit correlations that cannot be simulated classically.
For historical reasons such correlations are called ``Bell inequality violations.''
We give two new two-player games with Bell inequality violations
that are stronger, fully explicit, and arguably simpler than earlier work.

The first game is based on the Hidden Matching problem of quantum communication complexity,
introduced by Bar-Yossef, Jayram, and Kerenidis.
This game can be won with probability~1 by a strategy using
a maximally entangled state with local dimension $n$ (\eg, $\log n$ EPR-pairs),
while we show that the winning probability of any classical strategy differs
from ${1}/{2}$ by at most $O((\log n)/\sqrt{n})$.

The second game is based on the integrality gap for Unique Games by Khot and Vishnoi and the quantum rounding
procedure of Kempe, Regev, and Toner.  Here $n$-dimensional entanglement allows the game to be won with 
probability $1/(\log n)^2$, while the best winning probability without entanglement is $1/n$. 
This near-linear ratio is almost optimal, both in terms of the local dimension 
of the entangled state, and in terms of the number of possible outputs of the two players.
\end{abstract}

\tocacm{J.2} %
\tocams{81P68}

\tockeywords{quantum communication, nonlocal games, Bell inequalities, Fourier analysis}

\end{frontmatter}

\section{Introduction}

One of the most striking features of quantum mechanics is the fact that
\emph{entangled} particles can exhibit correlations that cannot be reproduced or
explained by classical physics, or more precisely, by ``local hidden-variable theories.''
This was first noted by Bell~\cite{bell:epr}
in response to Einstein, Podolsky, and Rosen's challenge to the completeness of quantum mechanics~\cite{epr}.
Experimental realization of such correlations is the strongest proof
we have that nature does not behave according to classical physics:
nature cannot simultaneously be ``local'' (meaning that information does not
travel faster than the speed of light) and ``realistic'' (meaning that measurable properties
of particles such as its spin always have a definite---if possibly unknown---value).
Many such experiments have been done. All behave in accordance with 
the predictions of quantum mechanics, though so far none has closed all ``loopholes'' that would allow
some (usually very contrived) classical explanation of the observations based
on imperfect behavior of, for instance, the photon detectors used.

Here we study quantitatively \emph{how much} such correlations obtained from entangled quantum systems
can deviate from what is achievable classically. It will be convenient to describe our results in terms of
two-player \emph{games}, which are described as follows.
Two non-communicating parties, called Alice and Bob, receive inputs $x$ and $y$ according
to some fixed and known probability distribution $\pi$, and are required to produce outputs $a$ and $b$, respectively.
There is a predicate specifying which outputs $a,b$ are correct on inputs $x,y$.
The definition of a game $G$ consists of this predicate and the distribution $\pi$. 
The goal is to design games where entangled strategies have much higher winning probability
than the best classical strategy.
While this setting is used to study non-locality in physics, the same set-up
is also used extensively to study the power of entanglement in computer science
contexts like multi-prover interactive proofs~\cite{kkmtv:entangledprovers,kkmv:usingent},
parallel repetition~\cite{csuu:parallelxorj,KRT08}, and cryptography.

Entangled strategies start out with an arbitrary fixed entangled
state.  No communication takes place between Alice and Bob.  For each
input $x$, Alice has a measurement, and for each input $y$, Bob has a
measurement.  They apply the measurements corresponding to $x$ and $y$
to their halves of the entangled state, producing classical outputs
$a$ and $b$, respectively.  Their goal is to maximize the winning
probability.  The \emph{entangled value} $\omega^*(G)$ of the game is
the supremum of the winning probability, taken over all entangled
strategies.  When restricting to strategies that use entanglement of
local dimension~$n$, the value is denoted $\omega^*_n(G)$.  This
should be contrasted with the \emph{classical value}
$\omega(G)=\omega^*_1(G)$ of the game, which is the maximum among all
classical, non-entangled strategies.  Shared randomness between the
two parties is allowed, but is easily seen not to be beneficial.

The remarkable fact, alluded to above, that some entanglement-based
correlations cannot be simulated classically, corresponds to the fact
that there are games $G$ for which the entangled value $\omega^*(G)$
is strictly larger than the classical value $\omega(G)$.  The CHSH
game is one particularly famous example~\cite{chsh}.  Here, the inputs
$x\in\bits$ and $y\in\bits$ are uniformly distributed, and Alice and Bob
win the game if their respective outputs $a\in\bits$ and $b\in\bits$
satisfy $a\oplus b=x\wedge y$; in other words, $a$ should equal $b$
unless $x=y=1$.  The classical value of this game is easily seen to be
$\omega(G)=3/4$, while the entangled value is known to be
$\omega^*(G)=1/2+1/(2\sqrt{2})\approx 0.85$.  The entangled value is
achieved already with 2-dimensional entanglement (\ie, one EPR-pair),
so $\omega^*(G)=\omega^*_2(G)$ for this game~\cite{tsirelson87}.

One common figure of merit of a game is that of the \emph{Bell inequality violation} exhibited by a game, which is defined as the ratio of entangled and classical values. More generally, we allow to replace the \emph{values} (which are the maximum winning probabilities) by \emph{biases} around some arbitrary ``center'' $p \in [0,1]$, where by bias we mean the maximum distance of the winning probability from $p$. For instance, by using $p=1/2$ as the center,
one can see that the CHSH game above exhibits a Bell inequality violation of $\sqrt{2}$.\footnote{Notice that this requires that the winning probability of non-entangled players is always between $1/4$ and $3/4$, which is clearly the case.}
In \expref{Section}{sec:formalbell} we explain the origin of the term ``Bell inequality violation,'' define it more formally, and explain the close relationship between games and Bell inequalities. 

In two recent papers, Junge et al.~\cite{PerezGarcia09arxiv,junge&palazuelos:largeviolation}
studied how large a Bell inequality violation one can obtain.
In terms of upper bounds, \cite{PerezGarcia09arxiv} proved that
the maximum Bell inequality violation $\omega^*_n(G)/\omega(G)$
obtainable with entangled strategies of local dimension $n$, is at most $O(n)$,
and~\cite[Theorem~6.8]{junge&palazuelos:largeviolation} proved that if Alice and Bob have
at most $k$ possible outputs each, then the violation $\omega^*(G)/\omega(G)$
is at most $O(k)$, irrespective of the amount of entanglement they can use.
(This improved an earlier $O(k^2)$ upper bound due to Degorre et al.~\cite{dklr:nonsignal},
and was also obtained for the special case of games by Dukaric~\cite[Theorem~4]{dukaric:norm}.)

In terms of lower bounds, \cite{PerezGarcia09arxiv} showed the existence of a Bell 
inequality violation of $\Omega(\sqrt{n}/(\log n)^2)$, 
where $n$ is both the entanglement dimension and the number of outputs of Alice and Bob.
This was improved to $\sqrt{n}/\log n$ in~\cite{junge&palazuelos:largeviolation}.
Both constructions are probabilistic, and the proofs show that with high probability
the constructed games exhibit a large violation, yielding the existence of such games, without giving an explicit formulation.
Their proofs are heavily based on connections to the mathematically beautiful areas of
Banach spaces and operator spaces, but as a result are arguably somewhat inaccessible to those unfamiliar
with these areas, and it is difficult to get a good intuition for them.
(It is actually possible to analyze their game and reprove many of their results---often with 
improved parameters---using elementary probabilistic techniques~\cite{regev:bell}.)

Our main result in this paper is to exhibit two fully explicit games with strong Bell inequality violations.
The first achieves the same violation as~\cite{junge&palazuelos:largeviolation}, namely $\sqrt{n}/\log n$, where $n$ is both the number of possible outputs and the dimension of entanglement.
The second achieves the much stronger violation of $n/(\log n)^2$, which is optimal up to a polylogarithmic factor by the results of~\cite{PerezGarcia09arxiv,junge&palazuelos:largeviolation}.
Even though the second game gives a much stronger violation, the first one still has some merit; for instance,
entanglement allows the players to win it with certainty. 
Interestingly, although addressing a question in mathematical physics,
both games are inspired by earlier work in theoretical computer science (communication complexity and unique games, 
respectively), and so is their analysis. 
In the remainder of this introduction we provide an overview of the two games,
followed by some discussion and comparison.

\subsection{The Hidden Matching game}

The ``Hidden Matching'' problem was introduced in quantum communication complexity
by Bar-Yossef et al.~\cite{bjk:q1wayj}, and many variants
of it were subsequently studied~\cite{gkrw:identificationj,gkkrw:1wayj,gavinsky:interactionvsnonlocality}.
The original version is as follows, where it should be noted that now we allow communication, in contrast to the setting of non-local games.
Let $n$ be a power of 2. Alice is given input $x \in\bits^n$ and Bob is given a perfect matching $M$,
\ie, a partition of the set $[n]=\{1,\ldots,n\}$ into $n/2$ disjoint pairs $\{i,j\}$.
Both inputs are uniformly distributed.\footnote{\label{notechangehm}All our results
also hold with minor modifications for the case that Bob's matching is chosen
uniformly from the set $\{M_k \mid k\in\{0,\ldots,n/2-1\}\}$, where the matching $M_k$ consists of the pairs $\{i,j\}$
where $i\leq n/2$ and $j=n/2+1+(i+k-1 \bmod n/2)$. This has the advantage of lowering
the number of possible inputs to Bob to $n/2$. The main thing is to notice that equation~\eqref{eqqmij} still holds
with respect to this new distribution on Bob's matching if we replace the right-hand side by $2/n$, and the rest of the proof goes through.}
We allow one-way communication from Alice to Bob, and Bob is required to output a pair $\{i,j\} \in M$ and a bit $v\in \bits$.
They win if $v = x_i \oplus x_j$.

In \expref{Section}{ssechmcomm} we show that if Alice sends Bob a $c$-bit message, then their optimal winning probability
is ${1}/{2}+\Theta({c}/{\sqrt{n}})$.
Bar-Yossef et al.~\cite{bjk:q1wayj} earlier proved this for $c=\Theta(\sqrt{n})$, using information theory.
However, their tools seem unable to give good bounds on the success probability for much smaller $c$.
Instead, the main mathematical tool we use in our analysis is the so-called ``KKL inequality''~\cite{kkl:influence}
from Fourier analysis of Boolean functions. Roughly speaking, this inequality implies that if the message that
Alice sends about $x$ is short, then Bob will not be able to predict the parity $x_i\oplus x_j$
well for many $\{i,j\}$ pairs. His matching $M$ is uniformly distributed, independent of $x$,
and contains only $n/2$ of all $\binom{n}{2}$ possible $\{i,j\}$ pairs.
Hence it is unlikely that he can predict any one of those $n/2$ parities well.
The KKL inequality was used before to analyze another variant of Hidden Matching
in~\cite{gkkrw:1wayj}, though their analysis is different and more complicated
because their variant of Hidden Matching is a promise problem with a non-product input distribution.

The following non-local version of the Hidden Matching problem (and the entangled strategy for it)
is originally due to Buhrman, and related problems were studied in~\cite{gkrw:identificationj,gavinsky:interactionvsnonlocality}.

\begin{definition}[Non-Local Hidden Matching Game ($\HM^{\textsc{nl}}_n$)]\label{defhmnl}
Let $n$ be a power of 2 and $\M_n$ be the set of all perfect matchings on the set $[n]$.
Alice is given $x \in \bits^n$ and Bob is given $M\in \M_n$, both distributed according to the uniform distribution.
Alice and Bob do not communicate.
Alice's output is a string $a \in \bits^{\log n} $ and Bob's output is an $\{i,j\} \in M$ and $d \in \bits$. 
They win the game if and only if
\begin{equation}\label{eq:HMnlcondition}
(a \cdot(i \oplus j)) \oplus d = x_i \oplus x_j\,,
\end{equation}
where the dot indicates inner product (modulo~2) of two $\log n$-bit strings.
\end{definition}

Observe that Alice has $n$ possible outputs~$a$ and Bob has $2\cdot n/2=n$ possible outputs~$(\{i,j\},d)$ given his matching.

A classical strategy that wins this game induces a protocol
for the original Hidden Matching problem with communication $c=\log n$ bits
and the same winning probability $p$, as follows. Alice sends Bob the $\log n$-bit output $a$ from
the non-local strategy, Bob computes $v=(a \cdot(i\oplus j))\oplus d$ and outputs $(\{i,j\},v)$. 
We have that $v=x_i\oplus x_j$ with probability $p$. 
Hence, our bound for the original communication problem implies that no classical strategy
can win with probability that differs from $1/2$ by more than $O({(\log n)}/{\sqrt{n}})$.

In contrast, there is a strategy that wins with probability~1 using $\log n$ EPR-pairs,
which shows $\omega^*_n(G)=1$.\footnote{The reader might be a bit confused by the seeming overloading of the meaning of `$n$'. 
Formally, `$n$' is a parameter in the specification of the game. As it happens, for both of our games it is also the number 
of possible outputs for each player, \emph{and} the local dimension of the entangled state that our strategy uses
(though we do not claim that this entanglement-dimension $n$ is necessary to achieve the best-possible entangled value).}
This game therefore exhibits a Bell violation of $\Omega(\sqrt{n}/\log n)$ (by measuring
the maximal deviation of the winning probability from $1/2$).
This order is the same as that obtained by Junge et al.~\cite{PerezGarcia09arxiv,junge&palazuelos:largeviolation},
but our game is fully explicit and arguably simpler
(which would help any future experimental realization).
One might feel that our reduction to a communication complexity lower bound is 
responsible for losing the $\log n$ factor; however in \expref{Theorem}{thm:lowerHMnl} 
we exhibit a classical strategy with winning probability $1/2 + \Omega(\sqrt{(\log n)/n})$. 
This shows that at least the square root of the log-factor is really necessary.

\subsection{The Khot-Vishnoi game}\label{ssec:kvgame}

Our second non-local game derives from the work of Khot and Vishnoi~\cite{khot&vishnoi:cutproblems}
on the famous \emph{Unique Games Conjecture} (UGC), which was introduced by Khot~\cite{khot:ugc}.
Although not necessary for the rest of this paper, we now provide some background and motivation. 
Roughly speaking, the UGC says that approximating the classical value of so-called unique games is 
a hard problem, even if we are only interested in a very rough approximation that can tell the difference
between value less than $\eps$ and value more than $1-\eps$. 
This conjecture implies many other hardness-of-approximation results
that do not seem obtainable using the more standard techniques based on the PCP theorem.
Khot and Vishnoi considered the standard semidefinite programming (SDP) relaxation of the classical value
and showed that there are games for which it provides a very poor approximation, in the sense
that the classical value is close to~$0$, yet the SDP relaxation is close to~$1$. This so-called integrality gap 
demonstrates that the standard SDP relaxation, which can be computed efficiently, does not lead to an algorithm
contradicting the UGC.

Kempe, Regev, and Toner~\cite{KRT08} already observed that they could
combine their ``quantum rounding'' technique with the game of~\cite{khot&vishnoi:cutproblems} to get
a game with $n$ possible outputs exhibiting a Bell inequality violation of $n^\eps$ for some small constant $\eps>0$, using entanglement
dimension $n$.
Our main contribution in the second part of this paper
is a refined (and at the same time simpler) analysis of both the Khot-Vishnoi game and of the quantum rounding technique.
We show that, somewhat surprisingly, nearly optimal violations can be obtained using this method.

We first give a precise definition of the Khot-Vishnoi (KV) game.

\begin{definition}[Khot-Vishnoi Game (KV$_n$)]\label{defkvgame}
The game is parametrized by an integer $n$, which we assume to be a power of~2, and a ``noise parameter'' $\eta\in[0,1/2]$. 
Consider the group $(\{0,1\}^n,\oplus)$ of $n$-bit strings together with bitwise addition mod~2, and let
$H$ be the subgroup containing the $n$ Hadamard codewords.\footnote{For $a\in\bits^{\log n}$, the corresponding $n$-bit Hadamard codeword is defined as $h(a):=(a\cdot j)_{j\in\bits^{\log n}}$.} This subgroup partitions $\{0,1\}^n$ into $2^n/n$ cosets of $n$ elements each. Alice receives a uniformly
random coset $x$ as input, which we can think of as $u\oplus H$ for uniformly random $u\in\bits^n$.
Bob receives a coset $y$ obtained from Alice's by adding a string of low Hamming weight,
namely $y=x\oplus z=u\oplus z\oplus H$, where each bit of $z\in\bits^n$ is set to~1 with probability $\eta$, independently of the other bits. 
Addition of $z$ gives a natural bijection between the two cosets, mapping each element of the first
coset to a relatively nearby element of the second coset; namely, the distance between the two elements
is the Hamming weight of~$z$, which is typically around $\eta n$. Each player is supposed to output one element
from its coset, and their goal is for their elements to match under the bijection.
In other words, Alice outputs an element $a \in x$, Bob outputs $b\in y$, and they win the game if and only if $a\oplus b=z$.%
\footnote{Note that the winning condition for this game is a ``randomized predicate,''
as there are $n$ possible ways to obtain the same $y$ from $x$, hence there are $n$ possible winning predicates (one for each $z\in x\oplus y$)
corresponding to each pair of inputs $x,y$. Strictly speaking, this requires a slightly more general definition of a game 
than the one given in the introduction; see the definition in \expref{Section}{sec:formalbell}. 
Although not relevant for any of our applications, we mention that one can 
modify the game in a straightforward manner, making it a game with a deterministic predicate.
The thing to observe is that with very high probability
exactly one of the $n$ predicates dominates, namely the one corresponding to a $z$ of Hamming
weight around $\eta n$.}
\end{definition}

Notice that the number of possible inputs to each player is $2^n/n$ and the number of possible outputs for each player is $n$.

Based on the integrality gap analysis of Khot and Vishnoi, 
in \expref{Section}{seckhotvishnoi} we show that no classical strategy can win this game with probability greater than $1/n^{\eta/(1-\eta)}$. 
We also sketch a classical strategy that achieves this winning probability up to lower order terms.
In contrast, using a simplified version of the ``quantum rounding'' technique of~\cite{KRT08},
we exhibit an entangled strategy that uses the $n$-dimensional maximally entangled state and wins with probability at least $(1-2\eta)^2$.
This strategy follows from the observation that each coset of $H$ defines an orthonormal basis of $\R^n$ in which
we can do a measurement.
Summarizing, we have entangled value $\omega^*_n(G)\geq (1-2\eta)^2$ and classical value $\omega(G)\leq 1/n^{\eta/(1-\eta)}$ for this game.
Setting the noise-rate to $\eta = 1/2 - 1/\log n$, the entangled value is $\Omega(1/(\log n)^2)$ while the classical value is $O(1/n)$, leading to a Bell inequality violation $\omega^*_n(G)/\omega(G)=\Omega(n/(\log n)^2)$.
Up to the polylogarithmic factor, this is optimal both in terms of the local dimension,
and in terms of the number of possible outputs.

Palazuelos~\cite{palazuelos:superactivation} recently used our result to prove an interesting \emph{super-activation} result.
He identified a constant-dimensional quantum state (a mixture of the maximally entangled state of local dimension~8 with the completely mixed state)
that cannot be used to violate any Bell inequality,
but a sufficiently large number of copies of which \emph{can} be used to violate a Bell inequality---namely the one associated with our KV game.

\subsection{Discussion and open problems}

Although Bell violations provide an elegant way to quantify the non-locality exhibited in a game, 
in experimental realizations of such games it is often important to take into account the actual classical and entangled values, and not just their ratio, especially when one tries to take into account possible imperfections in the experimental set-up.
The large Bell violation and the tiny success probability achievable by classical players seem to make the KV game attractive to an experimental realization.
One should keep in mind, though, that the success probability achievable by entangled players is $1/(\log n)^2$, which is somewhat low in absolute terms, and might not be visible if the experiment has too many false positives. It also means that an experiment must be repeated about $(\log n)^2$ times before we expect to see the first win. In the HM game, on the other hand, entangled players can win with certainty, which seems beneficial in case there are few false negatives. Another advantage of the HM game over KV is its somewhat simpler description.

One natural open question is to improve the Bell violation of $n/(\log n)^2$ achieved by the KV game 
either by tweaking the game or defining another game, possibly even matching the $O(n)$ upper bound up to a constant factor.%
\footnote{Interestingly, very recently Palazuelos~\cite{palazuelos:limitofmaxentstate} showed that this cannot 
be done using the maximally entangled state (which is the state used in all our entangled strategies): he proved that the ratio between the optimal value of entangled strategies 
using the maximally entangled state of local dimension~$n$, and the classical value, is at most $O(n/\sqrt{\log n})$.}
Throughout this paper we considered the Bell violation as a function of the number of outputs
of the players and/or of the dimension of entanglement. One can also analyze the violation
in terms of the number of possible \emph{inputs}. We recall that in the KV game both
players have inputs taken from an exponentially large set, and that in the HM
game (when modified as in \expref{footnote}{notechangehm}) Bob has only $n/2$ possible inputs,
but Alice still has an exponentially large set of inputs.
The Bell inequality violation of $\sqrt{n}/\log n$ presented by Junge and Palazuelos~\cite{junge&palazuelos:largeviolation}
has the advantage that the number of inputs is only $O(n)$.
Accordingly, another open question presents itself: can we find a game with a (near-)linear
Bell inequality violation, and linear number of inputs and outputs for both Alice and Bob?

Finally, while this paper focuses on the two-party setting,
obtaining stronger Bell inequality violations for settings with three or more parties is also a worthwhile goal.
P{\'e}rez-Garc{\'i}a et al.~\cite{PerezGarcia08} (see also~\cite{brietvidick11}) gave a randomized construction of a three-party \emph{XOR game} (in such a game each party outputs a bit, and winning or losing depends only on the XOR of those three bits)
that gives a Bell inequality violation of $\Omega(\sqrt{d})$ using an entangled state in dimensions $d\times D\times D$ with $D\gg d$.%
\footnote{They also showed that using GHZ states, there is no superconstant
Bell inequality violation for XOR games (see also~\cite{Briet2009blv}.)}
In contrast, it is a known consequence of Grothendieck's inequality that such non-constant separations
do not exist for \emph{two}-party XOR games.
We do not know how large Bell inequality violations can be for arbitrary three-party games.

Note that it is easy to make a three-party version of the Hidden Matching game: Alice gets input $x\in\bits^n$, Bob gets input $y\in\bits^n$, and Charlie gets a matching $M$ as input, all uniformly distributed.
The goal is that Alice outputs $a\in\bits^{\log n}$, Bob outputs $b\in\bits^{\log n}$,
Charlie outputs $d\in\bits$ and $\{i,j\}\in M$,
such that $((a\oplus b)\cdot (i\oplus j))\oplus d=x_i\oplus x_j\oplus y_i\oplus y_j$.
By modifying the two-party proofs in this paper, it is not hard to show that the winning probability
using an $n$-dimensional GHZ state is~1, while the best classical winning probability deviates
from $1/2$ by at most $O((\log n)^2/n)$.  So going from two to three parties squares the Bell inequality violation for Hidden Matching. This improvement unfortunately does not scale up with more than three parties, because one can show the classical winning probability is always at least $1/2+\Omega(1/n)$.

\section{Preliminaries}

\subsection{Fourier analysis}\label{secfourier}

The crucial technical tool used in the analysis of both of our games is Fourier analysis on the Boolean cube.  
We will just introduce what we need here, referring to~\cite{odonnell:survey,wolf:fouriersurvey} for more details and references.
Unless mentioned otherwise, expectations and probabilities are taken over a uniformly random $x\in\bits^n$.
Define the inner product between functions $f,g:\bits^n\rightarrow\mathbb{R}$ as
\[
\inpc{f}{g}=\frac{1}{2^n}\sum_{x\in\bits^n}f(x)g(x)=\E[f(x)\cdot g(x)]\,.
\]
For $S\subseteq[n]$, the function $\chi_S(x)=(-1)^{\sum_{i\in S}x_i}$ is the parity (represented as a sign) of the variables indexed in $S$.
These $2^n$ functions (one for each $S$) form an orthonormal basis for the space of all real-valued functions on the Boolean cube.
The \emph{Fourier coefficients} of $f$ are $\widehat{f}(S)=\inpc{f}{\chi_S}$, and we can write $f$ in
its Fourier decomposition
\[
f=\sum_{S\subseteq[n]}\widehat{f}(S)\chi_S\,.
\]
Since the $\chi_S$ form an orthonormal basis, it is easy to show
\begin{equation}\label{eq:plancherel}
\inpc{f}{g}=\sum_S\widehat{f}(S)\widehat{g}(S)\,.
\end{equation}
For $p \ge 1$, the \emph{$p$-norm} of $f$ is defined as
\[
\norm{f}_p=\E[|f(x)|^p]^{1/p}\,.
\]
This is monotone non-decreasing in $p$.
For $p=2$, equation~\eqref{eq:plancherel} with $g=f$ gives Parseval's identity: 
\[
\norm{f}_2^2=\sum_S\widehat{f}(S)^2\,.
\]
For $\rho \in [-1,1]$, the \emph{noise-operator} $T_{\rho}$ adds ``$\eta$-noise'' to each of the input bits, where $\eta=(1-\rho)/2$.
More precisely, the function $T_{1-2\eta} f$ is defined as  
\[
(T_{1-2\eta} f)(x)=\E_z[f(x\oplus z)]\,,
\] 
where $z\in\bits^n$ is an ``$\eta$-biased noise string,'' \ie, each bit $z_i$ is set to~1 with probability $\eta$, independently of the other bits.
The linear operator $T_\rho$ is diagonal in the Fourier basis:
it just multiplies each $\chi_S$ by the factor $\rho^{|S|}$.
Equivalently, 
\begin{equation}\label{eq:effectnoiseonfourier}
\widehat{T_\rho f}(S)=\rho^{|S|}\widehat{f}(S)\,.
\end{equation}
Since $T_\rho f$ is a convex combination of shifted copies $f_z(x)=f(x\oplus z)$ of $f$, 
by the triangle inequality we see that $T_\rho$ is a contraction:  $\norm{T_\rho f}_p\leq\norm{f}_p$ for all $p \ge 1$, $\rho \in [-1,1]$, and $f$.
The \emph{Bonami-Beckner hypercontractive inequality} implies $\norm{T_\rho f}_q\leq\norm{f}_p$ even for $q$ somewhat bigger than~$p$.

\begin{theorem}[Bonami-Beckner]\label{thbonamibeckner}
For $f:\bits^n\rightarrow\mathbb{R}$, $1\leq p\leq q$ and $\rho^2\leq(p-1)/(q-1)$, we have
\[
\norm{T_\rho f}_q\leq\norm{f}_p\,.
\]
\end{theorem}

An important consequence of the hypercontractive inequality is the so-called KKL inequality~\cite{kkl:influence}.

\begin{theorem}[KKL]\label{thkkl}
For $f:\bits^n\rightarrow\{-1,0,+1\}$ and $\delta\in[0,1]$, we have
\[
\sum_S\delta^{|S|}\widehat{f}(S)^2\leq \Pr[f(x)\neq 0]^{2/(1+\delta)}\,.
\]
\end{theorem}

\begin{proof}
Applying \expref{Theorem}{thbonamibeckner} with $q=2$, $p=1+\delta$, and $\rho=\sqrt{\delta}$
gives $\norm{T_\rho f}_2\leq\norm{f}_p$.  Note that because of the range of $f$, the right-hand side equals $\Pr[f(x)\neq 0]^{1/(1+\delta)}$.
Squaring both sides and rewriting the left-hand side using Parseval's identity completes the proof. 
\end{proof}

\subsection{A more formal look at Bell violations}\label{sec:formalbell}

Before we analyze the two games mentioned above, let us first say something more about
the mathematical treatment of general Bell inequalities.
Readers who are content with the above (more concrete) approach in terms of winning probabilities of games,
may safely skip this section.

Consider a game with $n$ possible inputs to each player and $k$ possible outputs.
The behavior of the players (irrespective of whether they use a classical or an entangled strategy)
can be summarized in terms of $n^2$ probability distributions, each on the set $[k] \times [k]$.
We denote by $P(ab\mid xy)$ the probability of producing outputs $a$ and $b$ when given inputs $x$ and $y$, 
with respect to a fixed strategy.
As described in the introduction, a game is defined by a probability distribution $\pi$ on
the input set $[n] \times [n]$, as well as a (possibly randomized) predicate on $[k] \times [k]$
for each input pair $(x,y)$.
The winning probability of the players can be written as
$$
\inpc{M}{P}=\sum_{abxy} M^{ab}_{xy}P(ab\mid xy)\,.
$$
where $M_{xy}^{ab}$ is defined as the probability of the input pair $(x,y)$ multiplied
by the probability that the output pair $(a,b)$ is accepted on this input pair.
We call $M=(M_{xy}^{ab})$ the \emph{Bell functional corresponding to the game}.
More generally, a \emph{Bell functional} is an arbitrary tensor $M=(M_{xy}^{ab})$
containing $n^2 k^2$ real numbers.

We define the \emph{classical value} of a Bell functional $M$ as
$$
\omega(M)=\sup_{P}|\inpc{M}{P}|\,,
$$
where the supremum is over all distributions $P$ representing classical strategies.
Similarly, the \emph{entangled value} of $M$ is defined as
$$
\omega^*(M)=\sup_{P}|\inpc{M}{P}|\,,
$$
where the supremum now is over all entangled strategies (using an entangled state of arbitrary dimension).
If the entangled state is restricted to local dimension $n$, the value is denoted $\omega^*_{n}(M)$.
We note that if $M$ is the Bell functional corresponding to a game, then these
definitions coincide with our definitions from the introduction, and in this case
the absolute value is unnecessary since $M$ is non-negative.

A \emph{Bell inequality} is an upper bound on $\omega(M)$ for some Bell functional $M$;
it shows a limitation of \emph{classical} strategies.\footnote{An upper bound on $\omega^*(M)$ is
known as a \emph{Tsirelson inequality}, and shows a limitation of entangled strategies.}
The \emph{Bell inequality violation} demonstrated by a Bell functional $M$ is defined as
the ratio between the entangled and the classical value
$$ 
\frac{\omega^*(M)}{\omega(M)}\,.
$$
This provides a convenient quantitative way to measure the extra power provided
by entanglement. This definition of Bell violation enjoys a rich mathematical structure,
as witnessed by the numerous connections found to Banach space and operator space theory~\cite{PerezGarcia09arxiv,junge&palazuelos:largeviolation,dukaric:norm}, and also has a beautiful geometrical interpretation as the ``distance'' between the set of
all classical strategies and the set of all entangled strategies. (See Section~6.1 in~\cite{junge&palazuelos:largeviolation}.)

Clearly, any game $G$ for which $\omega^*(G) \ge K \omega(G)$ gives a Bell violation of $K$ by just
taking the functional corresponding to $G$. But recall that in the introduction we said that one is also allowed 
to look at the ratio of biases around some center probability $p$. We now explain why this still agrees with
the above definition of Bell violation. We claim that if $G$ is a game for which the winning probability of
any classical strategy cannot deviate from $p$ by more than $\delta_1$ and, moreover,
there is an entangled strategy obtaining winning probability at least $p+\delta_2$ (or at most $p-\delta_2$),
then we obtain a Bell violation of $\delta_2/\delta_1$. To see why, let $M$ be the functional
corresponding to the game, and let $M'$ be the functional obtained by subtracting
from each $M_{xy}^{ab}$ the probability of input pair $(x,y)$ times $p$. Then it is
easy to see that for any strategy $P$, $\inpc{M'}{P}=\inpc{M}{P}-p$. Hence,
$\omega(M')$ and $\omega^*(M')$ are exactly the bias around $p$ of classical and entangled strategies, respectively,
and the claim follows. The converse to this statement is also true:
any Bell functional $M$ can be converted to a game in such a way that the ratio between the
entangled bias and the classical bias of the game (both around $1/2$, say) is exactly
the Bell violation demonstrated by the functional. To prove this, consider the game
in which an input pair $(x,y)$ is chosen uniformly, and outputs $a,b$ are accepted with probability
$1/2+\delta M_{xy}^{ab}$ for some sufficiently small $\delta>0$ so that all these probabilities are in $[0,1]$.

\section{Hidden Matching game}\label{sechiddenmatching}

In this section we define and analyze the Hidden Matching game.
Here and below, unless stated otherwise, all probabilities and expectations are taken over the distributions on $x$ and $M$ specified by the game.

\subsection{The Hidden Matching problem in communication complexity}\label{ssechmcomm}

While our focus is non-locality, it will actually
be useful to first study the original version of the Hidden Matching problem in the context of protocols
where communication from Alice to Bob is allowed.
Both the problem and the efficient quantum protocol below come from~\cite{bjk:q1wayj}. 

\begin{definition}[Hidden Matching ($\HM_n$)]
Let $n$ be a power of 2 and $\M_n$ be the set of all perfect matchings on the set $[n]$. Alice is given $x \in\bits^n$ and Bob is given $M\in \M_n$, both distributed according to the uniform distribution. We allow one-way communication from Alice to Bob, and Bob outputs an $\{i,j\} \in M$ and $v\in \bits$. They win if $v = x_i \oplus x_j$.
\end{definition}

\begin{theorem}\label{thm:quantumHM}
For every $n$ that is a power of~2, there is a protocol for $\HM_n$ with $\log n$ qubits of one-way communication that wins with probability~1
(\ie, $v = x_i \oplus x_j$ always holds).
\end{theorem}

\begin{proof}
The protocol is the following:
\begin{enumerate}
\item Alice sends Bob the state \(\displaystyle \ket{\psi} = \frac{1}{\sqrt{n}} \sum_{i=1}^n (-1)^{x_i} \ket{i} \).
\item Bob measures $|\psi\rangle$ in the $n$-element basis 
\[
B = \left\{ \frac{\ket{i} \pm \ket{j}}{\sqrt{2}}  \middle| \{i,j\} \in
M\right\}\,.
\]
\begin{itemize}
\item
If the outcome of the measurement is a state $\displaystyle
\frac{1}{\sqrt{2}}(\ket{i} + \ket{j})$, Bob outputs $\{i,j\}$ and
$v=0$.
\item
If the outcome of the measurement is a state $\displaystyle
\frac{1}{\sqrt{2}}(\ket{i} - \ket{j})$, Bob outputs $\{i,j\}$ and
$v=1$.
\end{itemize}
\end{enumerate}
For each $\{i,j\} \in M$, the probability to get $({1}/{\sqrt{2}})(\ket{i} + \ket{j})$ equals~$2/n$ if $x_i\oplus x_j = 0$,
and equals~0 otherwise,
and similarly for $({1}/{\sqrt{2}}) (\ket{i} - \ket{j})$ with the parity flipped. Hence Bob's output is always correct.
\end{proof}

\subsubsection{Limits of classical protocols for \texorpdfstring{$\HM_n$}{HM}}
Here we show that classical protocols with little communication cannot have good success probability.
To start, note that a protocol that uses shared randomness is just a probability distribution over
deterministic protocols, hence the maximal winning probability is achieved by a deterministic protocol.

\begin{theorem}\label{thm:upperHM}
Every classical deterministic protocol for $\HM_n$ with $c$ bits of one-way communication, where Bob outputs $(\{i,j\},v)$, has
\[
\Pr[ v = x_i \oplus x_j ] \leq \frac{1}{2} + \frac{c+1}{\sqrt{n-1}}\,.
\]
\end{theorem}

The intuition behind the proof is the following.
If the communication $c$ is small, the set~$X_m$ of inputs~$x$ for which Alice sends message $m$ will typically be large (of size about $2^{n-c}$), meaning Bob has little knowledge of most of the bits of $x$.
The KKL inequality implies that for most of the $\binom{n}{2}$ pairs $\{i,j\}$, Bob cannot guess the parity $x_i\oplus x_j$ well.  Of course, Bob has some freedom in which $\{i,j\}$ he outputs, but that freedom is limited to the $n/2$ pairs $\{i,j\}$ in his matching $M$, and it turns out that on average he will not be able to guess any of those parities well.

\begin{proof}
Fix a classical deterministic protocol.
For each $m\in\bits^c$, let $X_m \subseteq \bits^n$ be the set of Alice's inputs for which she sends message $m$.
These sets $X_m$ together partition Alice's input space $\bits^n$.
Define $p_m = |X_m|/2^n$.
Note that $\sum_m p_m = 1$, so $\{p_m\}$ is a probability distribution over the $2^c$ messages $m$.

For each $m$, define the following probability distribution over all possible pairs $\{i,j\}$: 
\[
q_m(\{i,j\}) = \Pr[\text{Bob outputs $\{i,j\}$} \mid \text{Bob received $m$}]\,.
\]
We have 
\begin{align}\label{eqqmij}
q_m(\{i,j\}) \leq \frac{1}{n-1}\,,
\end{align}
because we assume Bob always outputs an element from $M$ and for fixed $i\neq j$ we have $\Pr[\{i,j\}\in M]=1/(n-1)$,
since each $j$ is equally likely to be paired up with~$i$ under the uniform distribution on~$M$. This implies
\begin{align}\label{eqqmbound}
\sum_{\{i,j\}} q_m(\{i,j\})^2 \;\leq\; \max_{\{i,j\}} q_m(\{i,j\}) \cdot \sum_{\{i,j\}}q_m(\{i,j\}) \;=\; \max_{\{i,j\}} q_m(\{i,j\}) \leq \frac{1}{n-1}\,.
\end{align}
Define $\eps$ such that 
\[
\Pr[v = x_i \oplus x_j] = \frac{1}{2} + \eps\,,
\]
and
$\eps_m$ such that 
\[
\Pr[v = x_i \oplus x_j\mid \text{Bob received $m$}] = \frac{1}{2} +
\eps_m\,.
\]
Then $\eps = \sum_m p_m\eps_m$.
The best Bob can do when guessing $x_i\oplus x_j$ given message $m$,
is to output the value of $x_i\oplus x_j$ that occurs most often among the $x\in X_m$.
Define 
\[
\beta_{mij} = \E_{x\in X_m}[(-1)^{x_i\oplus x_j}]\,.
\]
Intuitively, if $X_m$ (and hence $p_m$) is large, then most of these $\beta_{mij}$ should be small.
We now use the KKL inequality (\expref{Theorem}{thkkl}) to make this intuition precise.

\begin{claim}\label{claim:KKLbeta}
$\displaystyle\sum_{\{i,j\}} \beta_{mij}^2 \leq
\begin{cases} 
4\log_2(1/p_m)^2 & \text{if $p_m\leq 1/2$\,,}\\
2 & \text{if $p_m>1/2$}\,.
\end{cases}$
\end{claim}

\begin{proof}
Let $f:\bits^n\rightarrow\bits$ be the characteristic function of the set $X_m$, so that $\Pr[f(x)\neq 0]=|X_m|/2^n=p_m$.
Observe that $\beta_{mij}$ is proportional to the Fourier coefficient $\widehat{f}(\{i,j\})$:
\[
\beta_{mij} \;=\; \E_{x\in X_m}[(-1)^{x_i\oplus x_j}] \;=\; \frac{2^n}{|X_m|} \E_{x\in\bits^n}[f(x)(-1)^{x_i\oplus x_j}]\;=\; \frac{1}{p_m}\widehat{f}(\{i,j\})\,.
\]
Using the KKL inequality with a $\delta\in[0,1]$ to be specified later, we get
\[
\delta^2\sum_{\{i,j\}} \widehat{f}(\{i,j\})^2
\;\leq\; \sum_{S\subseteq[n]}\delta^{|S|}\widehat{f}(S)^2\;\leq\; \Pr[f(x)\neq 0]^{2/(1+\delta)}\;=\;p_m^{2/(1+\delta)}\,,
\]
and therefore
\[
\sum_{\{i,j\}} \beta_{mij}^2 \;=\; \frac{1}{p_m^2}\sum_{\{i,j\}} \widehat{f}(\{i,j\})^2
   \;\leq\;  \frac{1}{\delta^2}(1/p_m)^{2-2/(1+\delta)}\,.
\]
For $p_m > 1/2$ simply choose $\delta=1$. For $p_m \le 1/2$, choose $\delta =1/\log_2(1/p_m)$, 
which is in $[0,1]$, so that the above is 
\[
\log_2(1/p_m)^2(1/p_m)^{2\delta/(1+\delta)}
  \;\leq\;   \log_2(1/p_m)^2(1/p_m)^{2\delta}
  \;=\;  4\log_2(1/p_m)^2\,. \qedhere
\]
\end{proof}

The fraction of $x\in X_m$ where $x_i\oplus x_j=0$ is $1/2 + \beta_{mij}/2$,
hence Bob's optimal success probability when guessing $x_i\oplus x_j$ is $1/2 + |\beta_{mij}|/2$.
This implies, for fixed $m$,
\[
\E_{\{i,j\}\sim q_m} \left[\frac{1}{2} + \frac{|\beta_{mij}|}{2}\right] \;\geq\; \Pr[v = x_i \oplus x_j] \;=\; \frac{1}{2} + \eps_m\,,
\]
where ``$\{i,j\}\sim q_m$'' means that $\{i,j\}$ is distributed according to distribution~$q_m$.
This allows us to upper bound $\eps_m$ for $m$ where $p_m\leq 1/2$:
\begin{align*}
 2\eps_m  \;& \leq\;  \E_{\{i,j\}\sim q_m}[|\beta_{mij}|]\\
 & =\;  \sum_{\{i,j\}} q_m(\{i,j\}) |\beta_{mij}|\\
 & \leq\;  \sqrt{\sum_{\{i,j\}}q_m(\{i,j\})^2} \cdot \sqrt{\sum_{\{i,j\}}\beta_{mij}^2}\\
 & \leq\;  \frac{2\log_2(1/p_m)}{\sqrt{n-1}}\,,
\end{align*}
where the second inequality is by Cauchy-Schwarz and the third uses both equation~\eqref{eqqmbound} and the first part of \expref{Claim}{claim:KKLbeta}.
Since the $p_m$ sum to~1, there can be at most one~$m$ for which $p_m>1/2$.
For that $m$ we have $\eps_m\leq 1/\sqrt{n-1}$ by an analogous argument combined with the second part of \expref{Claim}{claim:KKLbeta}.

Finally we can bound $\eps$, treating the (at most one) $m$ with $p_m>1/2$ separately:
\begin{align*}
 \eps \;& =\;  \sum_{m\in\bits^c} p_m\eps_m \\
 & \leq\;  \sum_m p_m\frac{\log_2(1/p_m)}{\sqrt{n-1}} + \frac{1}{\sqrt{n-1}}\\
 & =\;  \frac{1}{\sqrt{n-1}}\left( H(p)+1 \right)\\ 
 & \leq\;  \frac{c+1}{\sqrt{n-1}}\,,
\end{align*}
where $H(p)=\sum_m p_m\log_2(1/p_m)$ denotes the binary entropy function, 
which is at most $c$ since the distribution $\{p_m\}$ is on $2^c$ elements.
\end{proof}

\subsection{Classical protocol for \texorpdfstring{$\HM_n$}{HM}}

Here we design a classical protocol that achieves the above upper bound on the success probability.
This protocol has no bearing on the large Bell inequality violations that are
our main goal in this paper, but it is nice to know the previous upper bound on the maximal success probability is essentially tight.

\begin{theorem}\label{thm:lowerHM}
For every $n$ that is a power of~2, and every positive integer $c\leq \sqrt{n}$, there exists a classical protocol
for $\HM_n$ with $c$ bits of one-way communication, such that for all inputs $x,M$,
$$
\Pr [ v = x_i \oplus x_j ] = \frac{1}{2} + \Omega\left(\frac{c}{\sqrt{n}}\right)\,.
$$
\end{theorem}

\begin{proof}
Assume for simplicity that $\sqrt{n}$ is integer, and that $c$ is even and sufficiently large.
Alice and Bob use shared randomness to choose two disjoint subsets $S_1,S_2$ of $[n]$ of size $\sqrt{n}$ each.
Let $y$ denote the bits of $x$ located in the indices given by the first subset, and $z$ the
bits located in the indices given by the second subset.
Alice and Bob use shared randomness to produce $2^{c/2}$ random $\sqrt{n}$-bit strings $y^{(1)},\ldots,y^{(2^{c/2})}$.
For each $\ell$, the distance $d(y,y^{(\ell)})$ is distributed binomially, as the sum of $\sqrt{n}$ fair coin flips.

The following well-known fact about the tail of binomial distribution can be seen for instance by estimating 
$\binom{k}{k/2-\beta\sqrt{k}}$ using Stirling's approximation.

\begin{fact}\label{factbinomial}
There exists a universal constant $\gamma>0$ such that if $X$ is the sum of $k$ fair coin flips, 
then for all $0<\beta<\sqrt{k}/2$ we have 
\[
\Pr\bigl[X\leq k/2-\beta\sqrt{k}\bigr]\geq 2^{-\gamma(1+\beta^2)}\,.
\]
\end{fact}

Thus we have 
\[
\Pr\bigl[d(y,y^{(\ell)})\leq \sqrt{n}/2-\beta n^{1/4}\bigr]\geq 2^{-\gamma
  (1+\beta^2)}\,.
\]
Hence by choosing $\beta=\Theta(\sqrt{c})$, with probability close to~1, there will be an $\ell$ such that $y$ and $y^{(\ell)}$ are at relative distance
$\leq 1/2-\Omega(c^{1/2}/n^{1/4})$.
If so, Alice sends Bob the first such $\ell$, and otherwise she tells him there is no such $\ell$.
This costs $c/2$ bits of communication.
Similarly, at the expense of another $c/2$ bits of communication, Bob obtains an approximation of $z$
with relative distance at most $\leq 1/2-\Omega(c^{1/2}/n^{1/4})$.

It is easy to see that with probability at least 1/2, Bob's matching $M$ contains
an $\{i,j\}$ with $i\in S_1$ and $j\in S_2$.
Bob can predict $x_i$ with success probability $1/2+\Omega(c^{1/2}/n^{1/4})$ from his approximation of $y$,
and can predict $x_j$ with success probability $1/2+\Omega(c^{1/2}/n^{1/4})$ from his approximation of $z$.
These success probabilities are independent, hence he can predict $x_i\oplus x_j$ with success probability
$1/2+\Omega(c/\sqrt{n})$.
If there is no such $\{i,j\}\in M$, or if he did not get good approximations to~$y$ or~$z$,
then Bob just outputs any $\{i,j\}\in M$ and a random bit for $v$, giving success probability $1/2$.
Putting everything together, we have a protocol that wins with probability $1/2+\Omega(c/\sqrt{n})$.
\end{proof}

\subsection{Entangled value for \texorpdfstring{$\HM^{\textsc{nl}}_n$}{HMnl}}

We now port our results to the non-local setting, referring to \expref{Definition}{defhmnl} for 
the game associated with the Hidden Matching problem.

\begin{theorem} \label{thm:quantumHMnl}
For every $n$ that is a power of~2, there exists an entangled strategy for $\HM^{\textsc{nl}}_n$ using a maximally entangled state with local dimension $n$, such that condition~\eqref{eq:HMnlcondition} is always satisfied.
\end{theorem}

\begin{proof}
The strategy is as follows.
Alice and Bob share
\(\displaystyle \ket{\psi}= \frac{1}{\sqrt{n}} \sum_{i\in \bits^{\log n}}\ket{i}\ket{i}\).
\begin{enumerate}
\item Alice performs a phase-flip according to her input $x$. The state becomes
\[ 
\frac{1}{\sqrt{n}} \sum_{i\in \bits^{\log n}} (-1)^{x_i}\ket{i} \ket{i}\,.
\]
\item Bob performs a projective measurement with projectors $P_{ij} = \ketbra{i}{i} + \ketbra{j}{j}$, with $\{i,j\} \in M$. The state collapses to
 \[ \frac{1}{\sqrt{2}} \bigl[(-1)^{x_i}\ket{i} \ket{i} + (-1)^{x_j}\ket{j} \ket{j}\bigr]\]
for some $\{i,j\} \in M $ known to Bob.
\item Both players apply Hadamard transforms $H^{\otimes \log{n}}$, and the state becomes
\[
\frac{1}{\sqrt{2}n}\sum_{a,b \in \bits^{\log n}}\left( (-1)^{x_i+a\cdot i+b\cdot i} + (-1)^{x_j+a\cdot j+b\cdot j}\right)\ket{a}\ket{b}\,.
\]
\end{enumerate}
Notice that in the latter state, any pair $a,b$ with nonzero amplitude must satisfy that
$$ (a \cdot(i \oplus j)) \oplus (b \cdot(i \oplus j)) = x_i \oplus x_j\,. $$
Hence, if the players measure the state, Alice outputs $a$, and Bob outputs $\{i,j\}$ and the bit $d=b \cdot(i \oplus j)$, 
then they win the game with certainty.
\end{proof}

\subsection{Classical value for \texorpdfstring{$\HM^{\textsc{nl}}_n$}{HMnl}}

In contrast to the entangled value, the optimal classical winning probability is not much better than 1/2:

\begin{theorem} \label{thm:upperHMnl}
The winning probability of any classical strategy for $\HM^{\textsc{nl}}_n$ differs from ${1}/{2}$ by at most
$O\left((\log n)/{\sqrt{n}}\right)$.
\end{theorem}

\begin{proof}
A strategy that wins $\HM^{\textsc{nl}}_n$ with success probability $1/2+\eps$ can be turned into 
a protocol for $\HM_n$ with $\log{n}$ bits of communication and the same winning probability:
the players play $\HM^{\textsc{nl}}_n$, with Alice producing $a$ and Bob producing $\{i,j\},d$;
Alice then sends $a$ to Bob, who outputs $\{i,j\},(a \cdot(i\oplus j)) \oplus d$.
The latter bit equals $x_i\oplus x_j$ with probability $1/2+\eps$.
This requires $c=\log{n}$ bits of communication, so \expref{Theorem}{thm:upperHM} implies the upper
bound $1/2+O\left((\log n)/{\sqrt{n}}\right)$ on the winning probability. 
The lower bound of $1/2-O\left((\log n)/{\sqrt{n}}\right)$ on the winning probability follows similarly.
\end{proof}

Next we show that our upper bound on the success probability of classical strategies for $\HM^{\textsc{nl}}_n$ is nearly optimal: we can achieve advantage at least $\Omega(\sqrt{(\log n)/n})$. In %
the \exprefsilent{Appendix}{sec:alternativeproofhmnl}  %
we also give a simple alternative strategy with a slightly weaker advantage of $\Omega(1/\sqrt{n})$. The correctness of that more elementary strategy can be proven from first principles, unlike the one presented here. \expref{Theorem}{thm:lowerHMnl} and 
the \exprefsilent{Appendix}{sec:alternativeproofhmnl} 
have no bearing on our separation, but are included here mostly for 
completeness.

\begin{theorem}\label{thm:lowerHMnl}
For every $n$ that is a power of~2, there exists a classical deterministic strategy for $\HM^{\textsc{nl}}_n$ with winning probability
${1}/{2} + \Omega\big(\sqrt{(\log n)/{n}}\big)$ (under the uniform input distribution).
\end{theorem}

\begin{proof}
The strategy is as follows.
Bob finds the $j\in\{2,\ldots,n\}$ that is matched to~1 in $M$, and
outputs $\{1,j\}$ and $d=0$. Since the number~$i=1$ corresponds to the
string $0^{\log n}$, the winning condition $(a \cdot(i \oplus j))
\oplus d = x_i \oplus x_j$ is now equivalent to $a\cdot j=x_1\oplus
x_j$. Alice, given her input $x\in\bits^n$, outputs the
value~$a\in\bits^{\log n}$ that maximizes the winning probability
subject to $j$ being uniformly distributed over $\{2,\ldots,n\}$. That
is, she selects an~$a$ that maximizes the number 
\[
J_{ax}:=|\{j\in\{2,\ldots,n\} : a\cdot j=x_1\oplus x_j\}|\,.
\]
The winning probability of this strategy, for fixed $x$ and uniformly random~$M$, is $\max_a J_{ax}/(n-1)$. In the remainder of this proof we show that
\begin{equation}\label{eq:Jaxlowerbound}
\E_x[\max_a J_{ax}]\geq n/2+\Omega(\sqrt{n\log n})\,,
\end{equation}
which implies the theorem.

We use the following result due to Talagrand~\cite[Proposition 4.13]{ledoux&talagrand}.
For a finite set $T\subseteq\mathbb{R}^n$, define
\[
r(T):=\E_{z\in\{\pm 1\}^n}\left[ \sup_{t\in T}\left|\sum_{i=1}^n z_it_i\right| \right]\,.
\]
This is the expected maximal overlap between $z$ and the elements of $T$, expectation taken over uniformly random~$z\in\{\pm 1\}^n$. Let $N(T,d_2;\eps)$ denote the minimal number of (open) balls of radius $\eps>0$ (in Euclidean distance $d_2$) needed to cover~$T$. 
\begin{proposition}[Talagrand]\label{prop:ledouxtalagrand}
There exists a constant $K>0$ such that for any $\eps>0$ and $T \subseteq \mathbb{R}^n$, if $\max_{t\in T,i\in[n]}|t_i|\leq \eps^2/(Kr(T))$, then
\begin{equation*}
\eps\sqrt{\log N(T,d_2;\eps)}\leq K r(T)\,.
\end{equation*}
\end{proposition}
We apply this proposition as follows. For $a\in\bits^{\log n}$ consider the Hadamard codeword $h(a):=((-1)^{a\cdot j})_{j\in\bits^{\log n}}$, with bits represented as $\pm 1$ instead of $0/1$. Let $T=\{h(a) : a\in\bits^{\log n}\}$ be the set of $n$ Hadamard codewords. Any two distinct elements of $T$ differ in exactly $n/2$ positions, and hence are at Euclidean distance $\sqrt{2n}$. Therefore a ball of radius $\eps:=\sqrt{n}/2$ can contain at most one element of $T$, so we need exactly $n$ balls of radius~$\eps$ to cover~$T$ (\ie, $N(T,d_2;\eps)=n$). \expref{Proposition}{prop:ledouxtalagrand} now implies
\[
r(T)=\Omega(\sqrt{n\log n})\,.
\]
The definition of $r(T)$ takes the absolute value of $\sum_{i=1}^n z_it_i$, but by symmetry, with probability~1/2 this quantity is positive for the maximizing~$t$, and moreover the sum can never be smaller than $-\sqrt{n}$ since $T$ forms an orthogonal basis.  Hence we also have
\[
\E_{z\in\{\pm 1\}^n}\left[ \sup_{t\in T}\sum_{i=1}^n z_it_i \right] = \Omega(\sqrt{n\log n})\,.
\]
This means that we expect $z$ to be relatively close in Hamming distance to some Hadamard codeword: the expected number of positions where $z$ agrees with the closest $t\in T$ is $n/2+\Omega(\sqrt{n\log n})$. In our application, since $x$ itself is uniformly random, the string $y:=((-1)^{x_1\oplus x_j})_{j\in\bits^{\log n}}$ is uniformly random except for its first bit, which is fixed to~1. Our quantity $\max_a J_{ax}$ measures the number of positions where $y$ agrees with the closest $t\in T$, ignoring the first position. We conclude that
\[
\E_{x\in\bits^n}\left[ \max_{a\in\bits^{\log n}} J_{ax} \right] \geq n/2 + \Omega(\sqrt{n\log n}) - 1\,.
\]
This implies equation~\eqref{eq:Jaxlowerbound}.
\end{proof}

\section{The Khot-Vishnoi game}\label{seckhotvishnoi}

\subsection{The classical value}\label{seccupper}

In this section we analyze the classical value of the Khot-Vishnoi game (\expref{Definition}{defkvgame}).
Our main result is an upper bound on the classical value of $1/n^{\eta/(1-\eta)}$, 
based on the analysis from~\cite{khot&vishnoi:cutproblems}.

Before we give that upper bound, let us first argue that it is essentially tight, 
\ie, there exists a strategy whose winning probability is approximately $1/n^{\eta/(1-\eta)}$.
To get some intuition for this game, first think of $\eta$ as some small constant
(even though we will eventually choose it close to $1/2$), and consider the following natural
classical strategy:
\begin{quote}
Alice and Bob each output the element of their coset that has highest Hamming weight.
\end{quote}
The idea is that if $a$ is the element of highest Hamming weight in
Alice's coset $x$, we expect $a\oplus z$ to also be of high Hamming weight (because it is
close to $a$ in Hamming distance), and since $a\oplus z$ is in his coset $y$, 
Bob is somewhat likely to pick it as his output.
We now give a brief back-of-the-envelope calculation suggesting that the winning
probability of this strategy is approximately $1/n^{\eta/(1-\eta)}$; since it is not
required for our main result, we will not attempt to make this argument rigorous.

Let $t \ge 0$ be such that the probability that a binomial $B(n,1/2)$ variable is greater
than $(n+t)/2$ is $1/n$. Recalling that the cumulative distribution function of the binomial distribution $B(n,p)$ can be approximated
by that of the normal distribution $N(np, np(1-p))$, and that the probability that a normal variable
is greater than its mean by $s$ standard deviations is approximately $e^{-s^2/2}$,
we can essentially choose $t$ to be the solution to $e^{-t^2/(2n)} = 1/n$
(so $t=\sqrt{2n \ln n}$). Then we expect Alice's $n$-element coset to contain exactly one element 
of Hamming weight greater than $(n+t)/2$. Since the element $a$ that Alice picks is the one
of highest Hamming weight, we assume for simplicity that its Hamming weight
is $(n+t)/2$. The players win the game if and only if $a\oplus z$ has the highest weight
among Bob's $n$ elements, which we heuristically approximate by the event
that $a\oplus z$ has Hamming weight at least $(n+t)/2$. The Hamming weight of $a\oplus z$
is distributed as the sum of two independent binomial distributions $B((n+t)/2,1-\eta)$ and $B((n-t)/2,\eta)$, which
can be approximated as above by the normal distribution $N((n+t)/2 - \eta t, n\eta(1-\eta))$.
Hence, for the Hamming weight of $a\oplus z$ to be at least $(n+t)/2$, the normal
variable needs to be greater than its mean by $\eta t / \sqrt{n \eta(1-\eta)}$
standard deviations, and the probability of this happening is approximately
\[
e^{-\eta^2 t^2/(2 n \eta(1-\eta))} = \frac{1}{n^{\eta/(1-\eta)}}\,,
\]
as claimed.

Now we show that no classical strategy can be substantially better.
The main technical tool used in the proof is the hypercontractive inequality (\expref{Theorem}{thbonamibeckner}),
which is applicable to our setting because we choose $u$ uniformly and $u\oplus z$ 
may be viewed as a ``noisy version'' of $u$.

\begin{theorem}
For every $n$ that is a power of~2, and every $\eta\in[0,1/2]$, every classical strategy
for the Khot-Vishnoi game KV$_n$ (\expref{Definition}{defkvgame})
has winning probability at most $1/n^{\eta/(1-\eta)}$.
\end{theorem}

\begin{proof}
Recall that the inputs are generated as follows: we choose a uniformly random $u\in\bits^n$
and an $\eta$-biased $z\in\bits^n$, and define the respective inputs to be the cosets $u\oplus H$ and $u\oplus z\oplus H$.
We can assume without loss of generality that Alice and Bob's behavior is deterministic.
Define functions $A,B: \bits^n \to \bits$ by $A(u)=1$ if and only if Alice's output
on $u\oplus H$ is $u$, and similarly for Bob. Notice that by definition, these functions attain the
value $1$ on exactly one element of each coset, and hence $\E_u[A(u)]=\E_u[B(u)]=1/n$. 

Recall that the players win if and only if
the sum of Alice's output and Bob's output equals $z$. Hence for all $u,z$,
$\sum_{h \in H} A(u\oplus h) B(u\oplus z\oplus h)$ equals~$1$ if the players win on input pair $u\oplus H,u\oplus z\oplus H$,
and equals~$0$ otherwise. Therefore, the winning probability is given by
\begin{align*}
\E_{u,z}\biggl[\sum_{h \in H} A(u\oplus h) B(u\oplus z\oplus h)\biggr]&=
\sum_{h \in H}  \E_{u,z}\bigl[A(u\oplus h) B(u\oplus z\oplus h)\bigr] \\
&= n \E_{u,z}\bigl[A(u)B(u\oplus z)\bigr]\,,
\end{align*}
where the second equality uses the fact that for all $h$, $u\oplus h$ is uniformly distributed.

We use the Fourier analysis framework introduced in \expref{Section}{secfourier}.
We now have
\begin{align*}
\E_{u,z}[A(u)B(u\oplus z)] & =  \E_{u}[A(u)(T_{1-2\eta}B)(u)]\\
 & = \inpc{A}{T_{1-2\eta} B}\\[1ex]
 & = \inpc{T_{\sqrt{1-2\eta}}A}{T_{\sqrt{1-2\eta}}B}\\[1ex]
 & \leq  \norm{T_{\sqrt{1-2\eta}}A}_2\cdot\norm{T_{\sqrt{1-2\eta}}B}_2\\[1ex]
 & \leq  \norm{A}_{2-2\eta}\cdot\norm{B}_{2-2\eta}\\[1ex]
 & =  \Big(\E_u[A(u)]\Big)^{1/(2-2\eta)}\cdot\Big(\E_u[B(u)]\Big)^{1/(2-2\eta)}\\
 & =  \frac{1}{n^{1/(1-\eta)}}\,.
\end{align*}
Here the third equality follows from equations~\eqref{eq:plancherel} and~\eqref{eq:effectnoiseonfourier},
the first inequality is by Cauchy-Schwarz, the second is the hypercontractive inequality
(\expref{Theorem}{thbonamibeckner} with $q=2$, $p=2-2\eta$ and $\rho=\sqrt{1-2\eta}$) applied to each of $A$ and $B$, 
and the second to last equality uses the fact that $A,B$ have range~$\bits$.
We complete the proof by noting that $n/n^{1/(1-\eta)}=1/n^{\eta/(1-\eta)}$.
\end{proof}

\subsection{Lower bound on the entangled value}\label{secqlower}

In this section we describe a good entangled strategy for the Khot-Vishnoi game,
following the ideas of Kempe, Regev, and Toner~\cite{KRT08} and the SDP-solution of~\cite{khot&vishnoi:cutproblems}.

\begin{theorem}
For every $n$ that is a power of~2, and every $\eta\in[0,1/2]$, there exists an entangled strategy that wins KV$_n$
with probability at least $(1-2\eta)^2$, using a maximally entangled state of local dimension~$n$.
\end{theorem}

\begin{proof}
For $a\in\bits^n$, let $v^a\in\mathbb{R}^n$ denote the unit vector $((-1)^{a_i}/\sqrt{n})_{i\in[n]}$.
Notice that for all $a,b$, we have $\inpc{v^a}{v^b}=1-2d(a,b)/n$, where $d(a,b)$ denotes the Hamming distance between
$a$ and $b$. In particular, the $n$ vectors $v^a$, as $a$ ranges over a coset of $H$, form an orthonormal basis of $\R^n$.

The entangled strategy is as follows.
Alice and Bob start with the $n$-dimensional maximally entangled state.
Alice, given coset $x=u\oplus H$ as input, performs a projective measurement in the orthonormal basis
given by $\{v^a \mid a \in x\}$ and outputs the value $a$ given by the measurement.
Bob proceeds similarly with the basis $\{v^b \mid b \in y\}$ induced by his coset $y=x\oplus z\oplus H$.
A standard calculation now shows that the probability to obtain the pair of outputs $a,b$
is $\inpc{v^a}{v^b}^2/n$. Since the players win if and only if $b=a\oplus z$, the winning probability on inputs $x,y$ is given by
$$
\frac{1}{n}\sum_{a \in x} \inpc{v^a}{v^{a\oplus z}}^2
\;=\;\frac{1}{n}\sum_{a \in x} \left(1-\frac{2\,d(a,a\oplus z)}{n}\right)^2\;=\;\left(1-\frac{2|z|}{n}\right)^2\,,
$$
where $|z|$ denotes the Hamming weight 
of the $\eta$-biased string $z$.
Taking expectation and using convexity, the overall winning probability is
\[
\E_z\left[\left(1-\frac{2|z|}{n}\right)^2\right]\;\geq\;
\left(\E_z\left[1-\frac{2|z|}{n}\right]\right)^2\;=\;(1-2\eta)^2\,.
\qedhere \]
\end{proof}

\section{Appendix: An alternative strategy for \texorpdfstring{$\HM^{\textsc{nl}}_n$}{HMnl}}\label{sec:alternativeproofhmnl}

Here we give an alternative and slightly weaker version of \expref{Theorem}{thm:lowerHMnl}, with advantage $\Omega(1/\sqrt{n})$ instead of $\Omega(\sqrt{(\log n)/n})$.

\begin{proof}
Fix arbitrary inputs $x,M$. Bob always outputs $i=1$
and $j$
equal to
whatever is matched to $i$ by $M$. Consider the following two unit vectors
in $\R^n$,
$$
u = \left((-1)^{x_1 \oplus x_k}/\sqrt{n} \right)_{k\in [n]}\,,
\qquad \qquad v = e_j\,,
$$
where $e_j$ is the vector with $1$ in the $j$th coordinate and zero elsewhere.
Notice that Alice knows $u$, Bob knows $v$, and that $\inpc{u}{v} =
(-1)^{x_1 \oplus x_j}/\sqrt{n}$. The players use shared randomness to choose
a random unit vector $w \in \R^n$. Bob outputs $d=0$ if $\inpc{w}{v} > 0$,
and $d=1$ otherwise. Alice outputs $a=0^{\log n}$ if
$\inpc{w}{u} > 0$, and a uniform $a \in \bits^{\log n}$
otherwise.

We now analyze the success probability. Assume that $x_1 \oplus x_j = 0$ (the other case
being similar).
It is easy to see that the probability
of both $\inpc{w}{u}$ and $\inpc{w}{v}$ being positive is 
\[
\frac{1}{2} - \frac{1}{2\pi}\arccos \inpc{u}{v}\,,
\]
as this is essentially a two-dimensional question. They have the same probability
of both being negative, and probability $({1}/{2\pi}) \arccos \inpc{u}{v}$ to be in each
of the two remaining cases.
In the two cases that $\inpc{w}{u} \le 0$ (an event that happens with probability $1/2$),
$a \cdot (i \oplus j)$ is a uniform bit (since $i \neq j$) and the players win with
probability exactly $1/2$. Otherwise (\ie, if $\inpc{w}{u} > 0$), the players win if and only if $d = 0$
(\ie, if also $\inpc{w}{v}>0$).
Hence, using that $\arccos(z)=\pi/2-\Theta(z)$ for small $z$, the overall winning probability is
\[
\frac{1}{2} \cdot \frac{1}{2} + \frac{1}{2} - \frac{1}{2\pi}\arccos \inpc{u}{v} = \frac{1}{2} + \Theta\left(\frac{1}{\sqrt{n}}\right)\,.
\qedhere \]
\end{proof}

\paragraph{Acknowledgements.}
We thank Jop Bri\"et, Dejan Dukaric, Carlos Palazuelos, and Thomas Vidick for useful discussions and comments,
and Carlos also for sending us his manuscripts~\cite{palazuelos:superactivation,palazuelos:limitofmaxentstate}.  
We also thank the anonymous ToC referees and John Watrous for their many comments and suggestions, which improved the presentation of the paper.

\bibliographystyle{tocplain}   %
\bibliography{v008a027}

\providecommand{\bibhead}[1]{}
\expandafter\ifx\csname pdfbookmark\endcsname\relax%
  \providecommand{\tocrefpdfbookmark}{}
\else\providecommand{\tocrefpdfbookmark}{%
   \hypertarget{tocreferences}{}%
   \pdfbookmark[1]{References}{tocreferences}}%
\fi

\tocrefpdfbookmark
\begin{thebibliography}{10}

\bibitem{bjk:q1wayj}\bibhead{bjk:q1wayj}
{\sc Ziv Bar-Yossef, T.~S. Jayram, and Iordanis Kerenidis}: Exponential
  separation of quantum and classical one-way communication complexity.
\newblock {\em SIAM J. Comput.}, 38(1):366--384, 2008.
\newblock Preliminary version in
  \href{http://dx.doi.org/10.1145/1007352.1007379}{STOC'04}.
\newblock [\epfmtdoi{10.1137/060651835}]

\bibitem{bell:epr}\bibhead{bell:epr}
{\sc John~Stewart Bell}: On the {E}instein {P}odolsky {R}osen paradox.
\newblock {\em Physics}, 1:195--200, 1964.

\bibitem{Briet2009blv}\bibhead{Briet2009blv}
{\sc Jop Bri{\"e}t, Harry Buhrman, Troy Lee, and Thomas Vidick}: Multiplayer
  {XOR} games and quantum communication complexity with clique-wise
  entanglement.
\newblock {\em Quantum Inf. Comput.}, 2013.
\newblock To appear. Preprint (2009) at
  \href{http://arxiv.org/abs/0911.4007}{arXiv}.

\bibitem{brietvidick11}\bibhead{brietvidick11}
{\sc Jop Bri{\"e}t and Thomas Vidick}: Explicit lower and upper bounds on the
  entangled value of multiplayer {XOR} games.
\newblock {\em Communications in Mathematical Physics}, 2013.
\newblock To appear. Preprint (2011) at
  \href{http://arxiv.org/1108.5647}{arXiv}.
\newblock [\epfmtdoi{10.1007/s00220-012-1642-5}]

\bibitem{chsh}\bibhead{chsh}
{\sc John~F. Clauser, Michael~A. Horne, Abner Shimony, and Richard~A. Holt}:
  Proposed experiment to test local hidden-variable theories.
\newblock {\em Phys. Rev. Lett.}, 23(15):880--884, 1969.
\newblock [\epfmtdoi{10.1103/PhysRevLett.23.880}]

\bibitem{csuu:parallelxorj}\bibhead{csuu:parallelxorj}
{\sc Richard Cleve, William Slofstra, Falk Unger, and Sarvagya Upadhyay}:
  Perfect parallel repetition theorem for quantum {XOR} proof systems.
\newblock {\em Comput. Complexity}, 17(2):282--299, 2008.
\newblock Preliminary version in
  \href{http://dx.doi.org/10.1109/CCC.2007.24}{CCC'07}.
\newblock [\epfmtdoi{10.1007/s00037-008-0250-4}]

\bibitem{dklr:nonsignal}\bibhead{dklr:nonsignal}
{\sc Julien Degorre, Marc Kaplan, Sophie Laplante, and J{\'e}r{\'e}mie Roland}:
  The communication complexity of non-signaling distributions.
\newblock {\em Quantum Inf. Comput.}, 11(7{\&}8):649--676, 2011.
\newblock Preliminary version in
  \href{http://dx.doi.org/10.1007/978-3-642-03816-7_24}{MFCS'09}. Preprint
  (2008-11) at \href{http://arxiv.org/abs/0804.4859}{arXiv}.
\newblock [\epfmt{acm}{2230916.2230924}]

\bibitem{dukaric:norm}\bibhead{dukaric:norm}
{\sc Dejan~D. Dukaric}: The {H}ilbertian tensor norm and its connection to
  quantum information theory.
\newblock 2010.
\newblock [\epfmt{arxiv}{1008.1948v2}]

\bibitem{epr}\bibhead{epr}
{\sc Albert Einstein, Boris Podolsky, and Nathan Rosen}: Can quantum-mechanical
  description of physical reality be considered complete?
\newblock {\em Physical Review}, 47(10):777--780, 1935.
\newblock [\epfmtdoi{10.1103/PhysRev.47.777}]

\bibitem{gavinsky:interactionvsnonlocality}\bibhead{gavinsky:interactionvsnonlocality}
{\sc Dmitry Gavinsky}: Classical interaction cannot replace quantum
  nonlocality.
\newblock 2009.
\newblock [\epfmt{arxiv}{0901.0956}]

\bibitem{gkkrw:1wayj}\bibhead{gkkrw:1wayj}
{\sc Dmitry Gavinsky, Julia Kempe, Iordanis Kerenidis, Ran Raz, and Ronald
  de~Wolf}: Exponential separation for one-way quantum communication
  complexity, with applications to cryptography.
\newblock {\em SIAM J. Comput.}, 38(5):1695--1708, 2009.
\newblock Preliminary version in
  \href{http://dx.doi.org/10.1145/1250790.1250866}{STOC'07}.
\newblock [\epfmtdoi{10.1137/070706550}]

\bibitem{gkrw:identificationj}\bibhead{gkrw:identificationj}
{\sc Dmitry Gavinsky, Julia Kempe, Oded Regev, and Ronald de~Wolf}:
  Bounded-error quantum state identification and exponential separations in
  communication complexity.
\newblock {\em SIAM J. Comput.}, 39(1):1--24, 2009.
\newblock Preliminary version in
  \href{http://dx.doi.org/10.1145/1132516.1132602}{STOC'06}.
\newblock [\epfmtdoi{10.1137/060665798}]

\bibitem{junge&palazuelos:largeviolation}\bibhead{junge&palazuelos:largeviolation}
{\sc Marius Junge and Carlos Palazuelos}: Large violation of {B}ell
  inequalities with low entanglement.
\newblock {\em Communications in Mathematical Physics}, 306(3):695--746, 2011.
\newblock Preprint (2010) at \href{http://arxiv.org/abs/1007.3043v2}{arXiv}.
\newblock [\epfmtdoi{10.1007/s00220-011-1296-8}]

\bibitem{PerezGarcia09arxiv}\bibhead{PerezGarcia09arxiv}
{\sc Marius Junge, Carlos Palazuelos, David {P\'{e}rez-Garc{\'i}a}, Ignacio
  Villanueva, and Michael~M. Wolf}: Unbounded violations of bipartite {B}ell
  inequalities via operator space theory.
\newblock {\em Communications in Mathematical Physics}, 300(3):715--739, 2010.
\newblock Preprint (2009) at \href{http://arxiv.org/abs/0910.4228}{arXiv}.
  Shorter version appeared in
  \href{http://dx.doi.org/10.1103/PhysRevLett.104.170405}{Phys. Rev. Lett.}
\newblock [\epfmtdoi{10.1007/s00220-010-1125-5}]

\bibitem{kkl:influence}\bibhead{kkl:influence}
{\sc Jeff Kahn, Gil Kalai, and Nathan Linial}: The influence of variables on
  {B}oolean functions.
\newblock In {\em Proc. 29th FOCS}, pp. 68--80. IEEE Comp. Soc. Press, 1988.
\newblock [\epfmtdoi{10.1109/SFCS.1988.21923}]

\bibitem{kkmtv:entangledprovers}\bibhead{kkmtv:entangledprovers}
{\sc Julia Kempe, Hirotada Kobayashi, Keiji Matsumoto, Ben Toner, and Thomas
  Vidick}: Entangled games are hard to approximate.
\newblock {\em SIAM J. Comput.}, 40(3):848--877, 2011.
\newblock Preliminary version in
  \href{http://dx.doi.org/10.1109/FOCS.2008.8}{FOCS'08}.
\newblock [\epfmtdoi{10.1137/090751293}]

\bibitem{kkmv:usingent}\bibhead{kkmv:usingent}
{\sc Julia Kempe, Hirotada Kobayashi, Keiji Matsumoto, and Thomas Vidick}:
  Using entanglement in quantum multi-prover interactive proofs.
\newblock {\em Comput. Complexity}, 18(2):273--307, 2009.
\newblock Preliminary version in
  \href{http://dx.doi.org/10.1109/CCC.2008.6}{CCC'08}.
\newblock [\epfmtdoi{10.1007/s00037-009-0275-3}]

\bibitem{KRT08}\bibhead{KRT08}
{\sc Julia Kempe, Oded Regev, and Ben Toner}: Unique games with entangled
  provers are easy.
\newblock {\em SIAM J. Comput.}, 39(7):3207--3229, 2010.
\newblock Preliminary version in
  \href{http://dx.doi.org/10.1109/FOCS.2008.9}{FOCS'08}. Preprint (2007-09) at
  \href{http://arxiv.org/abs/0710.0655}{arXiv}.
\newblock [\epfmtdoi{10.1137/090772885}]

\bibitem{khot:ugc}\bibhead{khot:ugc}
{\sc Subhash Khot}: On the power of unique 2-prover 1-round games.
\newblock In {\em Proc. 34th STOC}, pp. 767--775. ACM Press, 2002.
\newblock [\epfmtdoi{10.1145/509907.510017}]

\bibitem{khot&vishnoi:cutproblems}\bibhead{khot&vishnoi:cutproblems}
{\sc Subhash~A. Khot and Nisheeth~K. Vishnoi}: The unique games conjecture,
  integrality gap for cut problems and embeddability of negative type metrics
  into {$\ell_1$}.
\newblock In {\em Proc. 46th FOCS}, pp. 53--62. IEEE Comp. Soc. Press, 2005.
\newblock [\epfmtdoi{10.1109/SFCS.2005.74}]

\bibitem{ledoux&talagrand}\bibhead{ledoux&talagrand}
{\sc Michel Ledoux and Michel Talagrand}: {\em Probability in {B}anach Spaces}.
\newblock Springer, 1991.

\bibitem{odonnell:survey}\bibhead{odonnell:survey}
{\sc Ryan O'Donnell}: Some topics in analysis of {B}oolean functions.
\newblock Technical Report 055, 2008.
\newblock Available at \href{http://eccc.hpi-web.de/report/2008/055/}{ECCC}.
  Paper for invited talk at
  \href{http://dx.doi.org/10.1145/1374376.1374458}{STOC'08}.

\bibitem{palazuelos:limitofmaxentstate}\bibhead{palazuelos:limitofmaxentstate}
{\sc Carlos Palazuelos}: Bounding the largest {B}ell violation attainable by a
  quantum state.
\newblock 2012.
\newblock [\epfmt{arxiv}{1206.3695}]

\bibitem{palazuelos:superactivation}\bibhead{palazuelos:superactivation}
{\sc Carlos Palazuelos}: Superactivation of quantum nonlocality.
\newblock {\em Phys. Rev. Lett.}, 109(19):190401, Nov 2012.
\newblock Preprint (2012) at \href{http://arxiv.org/abs/1205.3118}{arXiv}.
\newblock [\epfmtdoi{10.1103/PhysRevLett.109.190401}]

\bibitem{PerezGarcia08}\bibhead{PerezGarcia08}
{\sc David {P{\'e}rez-Garc{\'i}a}, Michael~M. Wolf, Carlos Palazuelos, Ignacio
  Villanueva, and Marius Junge}: Unbounded violations of tripartite {B}ell
  inequalities.
\newblock {\em Communications of Mathematical Physics}, 279:455, 2008.
\newblock Preprint (2012) at
  \href{http://arxiv.org/abs/quant-ph/0702189}{arXiv}.
\newblock [\epfmtdoi{10.1007/s00220-008-0418-4}]

\bibitem{regev:bell}\bibhead{regev:bell}
{\sc Oded Regev}: Bell violations through independent bases games.
\newblock {\em Quantum Inf. Comput.}, 12(1-2):9--20, 2012.
\newblock Preprint (2011) at \href{http://arxiv.org/abs/1101.0576}{arXiv}.
\newblock [\epfmt{acm}{2231038}]

\bibitem{tsirelson87}\bibhead{tsirelson87}
{\sc Boris~S. Tsirel'son}: Quantum analogues of the {B}ell inequalities. {T}he
  case of two spatially separated domains.
\newblock {\em J. Math. Sciences}, 36(4):557--570, 1987.
\newblock [\epfmtdoi{10.1007/BF01663472}]

\bibitem{wolf:fouriersurvey}\bibhead{wolf:fouriersurvey}
{\sc Ronald~de Wolf}: A brief introduction to {F}ourier analysis on the
  {B}oolean cube.
\newblock {\em Theory of Computing Library, Graduate Surveys}, 1:1--20, 2008.
\newblock [\epfmtdoi{10.4086/toc.gs.2008.001}]

\end{thebibliography}

\begin{tocauthors}
\begin{tocinfo}[buhrman]
 Harry Buhrman\\
 Professor\\
 CWI and University of Amsterdam\\
 buhrman\tocat{}cwi\tocdot{}nl\\
 \url{http://homepages.cwi.nl/~buhrman}
\end{tocinfo}

\begin{tocinfo}[regev]
  Oded Regev\\
  Professor\\
  Blavatnik School of Computer Science, Tel Aviv University, and CNRS, ENS Paris.\\
  \url{http://www.cs.tau.ac.il/~odedr}
\end{tocinfo}

\begin{tocinfo}[scarpa]
 Giannicola Scarpa\\
 Ph.\,D. student\\
 CWI Amsterdam\\
 g\tocdot{}scarpa\tocat{}cwi\tocdot{}nl\\
 \url{http://homepages.cwi.nl/~scarpa}
\end{tocinfo}

\begin{tocinfo}[dewolf]
 Ronald de Wolf\\
 Professor\\
 CWI and University of Amsterdam\\
 rdewolf\tocat{}cwi\tocdot{}nl\\
 \url{http://homepages.cwi.nl/~rdewolf}
\end{tocinfo}

\end{tocauthors}

\begin{tocaboutauthors}
\begin{tocabout}[buhrman]
\textsc{Harry Buhrman} received his \phd\ from the \href{http://www.english.uva.nl/start.cfm}{University of Amsterdam} in 1993.  His advisors
were \href{http://staff.science.uva.nl/~peter/}{Peter van Emde Boas} and \href{http://www.cs.bu.edu/~homer/}{Steven Homer}.  
He is now group leader of the Algorithms and Complexity group at CWI and holds a part-time position as full professor at the University of Amsterdam.
His interests include quantum computing, complexity theory, and computational biology.  
\end{tocabout}
\begin{tocabout}[regev]
\textsc{Oded Regev} graduated from
  \href{http://www.tau.ac.il/}{Tel Aviv University} in 2001 under the
  supervision of \href{http://www.cs.tau.ac.il/~azar/}{Yossi Azar}.
  He spent two years as a postdoc at the \href{http://www.ias.edu/}{Institute for
    Advanced Study}, Princeton, and one year at the
  \href{http://www.berkeley.edu/}{University of California, Berkeley}.
  He is currently with the cryptography group at the {\'E}cole Normale Sup{\'e}rieure, Paris.  %
  His research interests include quantum computation, computational
  aspects of lattices, and other topics in theoretical computer
  science. He also enjoys photography, especially of his baby girl.
\end{tocabout}
\begin{tocabout}[scarpa]
  \textsc{Giannicola Scarpa} is a \phd\ student at
  \href{http://www.cwi.nl}{CWI}, Amsterdam, supervised by Ronald de
  Wolf.  In 2009, he received a Master's degree in Computer Science
  from the \href{http://www.unisa.it}{University of Salerno},
  Italy. His research interests include quantum computing,
  non-locality, combinatorial optimization, and game theory.  In his
  free time, he is a devoted movie-goer but unsuccessful movie
  maker, he devours short stories, writes some, and he often claims
  he is going to lose weight.
\end{tocabout}
\begin{tocabout}[dewolf]
\textsc{Ronald de Wolf} received his \phd\ from the \href{http://www.english.uva.nl/start.cfm}{University of Amsterdam} and \href{http://www.cwi.nl/}{CWI} in 2001.  His advisors
were \href{http://homepages.cwi.nl/~buhrman/}{Harry Buhrman} and \href{http://homepages.cwi.nl/~paulv/}{Paul Vit\'{a}nyi}.  
After doing a postdoc at the
  \href{http://www.berkeley.edu/}{University of California, Berkeley}, he now holds a permanent position at CWI and 
a part-time position as full professor at the University of Amsterdam.
His CS interests include quantum computing, complexity theory, and learning theory.  
He also holds a degree in philosophy, and enjoys classical music and literature.
\end{tocabout}

\end{tocaboutauthors}

\end{document}